\documentclass[aps,showpacs,amssymb,amsfonts,superscriptaddress,twocolumn,prl,
nofootinbib]{revtex4-1}

\usepackage{bbm,mathrsfs}
\usepackage{graphicx}
\usepackage{amsfonts,amsmath}
\usepackage{amsthm}
\usepackage{times}

\def\textbf#1{{\bf #1}}
\def\be{\begin{equation}}
\def\ee{\end{equation}}
\def\ben{\begin{eqnarray}}
\def\een{\end{eqnarray}}
\def\eea{\end{array}}
\def\bea{\begin{array}}
\newcommand{\Tr}[0]{\mathrm{Tr}}
\newcommand{\ot}[0]{\otimes}
\newcommand{\bei}{\begin{itemize}}
\newcommand{\eei}{\end{itemize}}
\newcommand{\ket}[1]{|#1\rangle}

\usepackage{changes}

\definecolor{myurlcolor}{rgb}{0,0,0.7}
\definecolor{myrefcolor}{rgb}{0.8,0,0}
\usepackage{hyperref}
\hypersetup{colorlinks, linkcolor=myrefcolor,
citecolor=myurlcolor, urlcolor=myurlcolor}

\newtheorem{thm}{Theorem}

\theoremstyle{definition}

\begin{document}

\title{Detecting non-locality in multipartite quantum systems with two-body correlation functions}

\author{J. Tura}
\affiliation{ICFO--Institut de Ciencies Fotoniques, 08860
Castelldefels (Barcelona), Spain}

\author{R. Augusiak}
\affiliation{ICFO--Institut de Ciencies Fotoniques, 08860
Castelldefels (Barcelona), Spain}

\author{A. B. Sainz}
\affiliation{ICFO--Institut de Ciencies Fotoniques, 08860
Castelldefels (Barcelona), Spain}

\author{T. V\'ertesi}
\affiliation{Institute of Nuclear Research of the Hungarian Academy of Sciences
H-4001 Debrecen, P.O. Box 51, Hungary}

\author{M. Lewenstein}
\affiliation{ICFO--Institut de Ciencies Fotoniques, 08860
Castelldefels (Barcelona), Spain} \affiliation{ICREA--Institucio
Catalana de Recerca i Estudis Avan\c{c}ats, Lluis Companys 23,
08010 Barcelona, Spain}

\author{A. Ac\'in}
\affiliation{ICFO--Institut de Ciencies Fotoniques, 08860
Castelldefels (Barcelona), Spain}
\affiliation{ICREA--Institucio
Catalana de Recerca i Estudis Avan\c{c}ats, Lluis Companys 23,
08010 Barcelona, Spain}

\begin{abstract}
Bell inequalities define experimentally observable quantities to
detect non-locality. In general, they involve correlation
functions of all the parties. Unfortunately, these measurements
are hard to implement for systems consisting of many constituents,
where only few-body correlation functions are accessible. Here we
demonstrate that higher-order correlation functions are not
necessary to certify nonlocality in multipartite quantum states by
constructing Bell inequalities from one- and two-body correlation
functions for an arbitrary number of parties. The obtained
inequalities are violated by some of the Dicke states, which arise
naturally in many-body physics as the ground states of the
two-body Lipkin-Meshkov-Glick Hamiltonian.
\end{abstract}

 \pacs{}

\maketitle

Local measurements on entangled composite quantum systems may lead
to correlations that cannot be simulated by any local
deterministic strategy assisted by shared randomness~
\cite{bell64,NLreview}. This phenomenon is known as nonlocality.
Apart from its fundamental interest, non-locality has also
turned into a key resource for certain
information-theoretic tasks, such as key distribution~\cite{key}
or certified quantum randomness generation~\cite{rand}.
Hence, revealing the nonlocality of a given
composite quantum state, or, in other words, certifying that it
can be used to generate nonlocal correlations upon local
measurements, is one of the central problems of quantum
information theory.

Due to the structure of the set of classical correlations (see
below), the natural way of tackling this problem is
to use Bell inequalities \cite{bell64}. These are linear
inequalities formulated in terms of expectation values
(correlators) of tensor products of measurements performed by the
observers, and their violation signals nonlocality. Many
constructions of Bell inequalities have been proposed (see e.g.
Refs. \cite{nierownosci}), however, most of them involve
full-order correlators, that is, expectation values of observables
of all parties). Intuitively, the latter
carry most of the information about correlations, and consequently
Bell inequalities based on them are the strongest ones, or even
tight (see, e.g., Refs. \cite{Sliwa,JD1}). But are these all-partite mean values necessary to reveal nonlocality? It was recently shown in Ref.
\cite{BSV,Wiesniak} that this is not the case, although expectation values with all but one parties are still involved. Hence, one is led to the more demanding question of whether certification of non-locality is possible from the minimal information achievable in this type of experiments, i.e. two-body expectation values.

This question also arises naturally in the context of experimental
implementations of Bell tests, in particular, in multipartite
systems. It should be stressed that several interesting
multipartite states are already within reach of current
experimental technology. In particular, four-qubit Smolin state
\cite{expSmolin}, eight-qubit GHZ state \cite{expGHZ}, and various
Dicke states \cite{DickeExp,DickeExp2} were experimentally
generated. However, in the case of large systems determining
experimentally expectation values of high-order is a hard task.
Designing nonlocality tests that rely solely on low-order
correlators would facilitate their experimental implementation.

It should also be stressed that an analogous question was already
explored in the case of entanglement, which next to nonlocality is
a key resource of quantum information theory
\cite{Horodeckireview}. Several entanglement criteria relying
solely on two-body expectation values have been proposed
\cite{EntTwoBody}. In particular, in \cite{coll} the possibility
of adressing two-body statistics (although not individually)
\textit{via} collective observables was exploited.

In this letter we address the above question and propose a class
of Bell inequalities constructed only from one and two-body
correlators that are violated by quantum states. We simplify the
problem by considering a subclass of symmetric Bell inequalities,
i.e., those that are invariant under a swap of any pair of parties
and characterize the corresponding polytope of classical
correlations. We also show that our inequalities are powerful
enough to certify nonlocality of the Dicke states that are ground
states of the two-body Lipkin-Meshkov-Glick Hamiltonian \cite{LMG}, making
our results promising from the experimental point of view.

\textit{Preliminaries.--}Let us consider the standard Bell-type
experiment in which $N$ spatially separated observers perform
measurements on their shares of some $N$-partite composite quantum
state $\rho$. In what follows we focus on
the simplest case where each party freely chooses
one between two dichotomic measurements, whose
outcomes we denote $\pm1$. A convenient way of
describing the established correlations in the two-outcome case is
to use the collection of expectation values (also called
correlators)
\begin{equation}\label{correlators}
\{\langle\mathcal{M}_{j_1}^{(i_1)}\ldots
\mathcal{M}^{(i_k)}_{j_k}\rangle\;|\; k=1,\ldots,N\}\qquad
\end{equation}
with $i_l=1,\ldots,N$ and $j_l=0,1$
$(l=1,\ldots,k)$.
We will refer to these collections as to ordered real vectors of
dimension $3^N-1$, and by saying correlations we mean
the corresponding vector. Also, the \textit{order} of
a correlator is the number of parties $k$ it involves [cf Eq.
(\ref{correlators})], and, in particular, those with $k=N$ we call the
\textit{highest-order} correlators, while those with $k=2$ the lowest-order or two-body correlators.

Within this framework, we say that the correlations represented by
(\ref{correlators}) are \textit{classical} (or \textit{local})
whenever, even if obtained
from composite quantum states, they can
be simulated by the observers with some shared classical
information as the only resource. Such correlations form a
polytope $\mathbbm{P}$, whose vertices are those collections
(\ref{correlators}) in which every correlator takes the product
form $\langle\mathcal{M}_{j_1}^{(i_1)}\ldots
\mathcal{M}^{(i_k)}_{j_k}\rangle=\langle
\mathcal{M}_{j_1}^{(i_1)}\rangle\cdot\ldots\cdot \langle
\mathcal{M}_{j_k}^{(i_k)}\rangle$ with individual mean values
$\langle \mathcal{M}_{j_1}^{(i_1)}\rangle$ being $\pm1$.

Bell was the first to recognize that the set of classical
correlations can be constrained by certain inequalities,
referred to as Bell inequalities \cite{bell64}. In fact, since
classical correlations form a polytope, $\mathbbm{P}$ can be fully
determined by a finite number of \textit{tight} Bell inequalities,
i.e., those corresponding to the facets of $\mathbbm{P}$.
Correlations that fall outside of $\mathbbm{P}$ are called
nonlocal. Consequently, the problem of characterizing all
classical correlations reduces to finding all tight Bell
inequalities for a given scenario. And, even if it sounds simple,
the problem is difficult to resolve as the number of facets
of the local polytope grows rapidly with the number of parties.

\textit{Bell inequalities from one- and two-body correlators.--}Most of the known
constructions of multipartite Bell inequalities contain
highest-order correlators, i.e., those with $k=N$ in Eq.
(\ref{correlators}).
In the following, we will see that 
one can design Bell inequalities that
witness nonlocality only from one and two-body expectation values.
A general form of such a Bell inequality is
\begin{eqnarray}\label{2corrBI}
&&\hspace{-0.4cm}\sum_{i=1}^{N}(\alpha_i\langle \mathcal{M}^{(i)}_0\rangle
+\beta_i\langle\mathcal{M}^{(i)}_1\rangle)+\sum_{i<j}^{N}\gamma_{ij}
\langle \mathcal{M}_{0}^{(i)}\mathcal{M}_{0}^{(j)}\rangle+\nonumber\\
&&\hspace{-0.4cm}+\sum_{i\neq j}^{N}\delta_{ij}
\langle
\mathcal{M}_{0}^{(i)}\mathcal{M}_{1}^{(j)}\rangle+\sum_{i<j}^{N}\varepsilon_{ij}
\langle \mathcal{M}_{1}^{(i)}\mathcal{M}_{1}^{(j)}\rangle+
\beta_C\geq 0,
\end{eqnarray}
where $\alpha_i,\beta_j,\gamma_{ij},\delta_{ij}$, and
$\varepsilon_{ij}$ are some real parameters, while $\beta_C$ is the
so-called classical bound. The corresponding polytope
$\mathbbm{P}_2$ of classical correlations is one constructed from
the elements of $\mathbbm{P}$ by neglecting correlators of order
higher than two. In other words, we take all elements (vectors) of
$\mathbbm{P}$ and simply remove those with $k\geq 3$ [cf. Eq.
(\ref{correlators})]. Analogously, the vertices of $\mathbbm{P}_2$
are those collections of correlators for which
$\langle\mathcal{M}_{k}^{(i)}\mathcal{M}_{l}^{(j)}\rangle=\langle\mathcal{M}_{k}
^{(i)}\rangle\cdot\langle\mathcal{M}_{l}^{(j)}\rangle$, while the
individual mean values are $\pm1$.

The characterization of $\mathbbm{P}_2$ reduces to finding all its
facets, i.e., \textit{tight two-body Bell inequalities}. Although
$\dim\mathbbm{P}_2=2N^2$ is much smaller than the one of
$\mathbbm{P}$, $3^N-1$, it still grows with $N$, thus
difficulting the task of determining facets of
$\mathbbm{P}_2$. One way to overcome this
problem (and keep the dimension constant irrespectively of $N$) is
to consider Bell inequalities that obey some symmetries. For
instance, one could consider translationally invariant Bell
inequalites consisting of correlators involving only nearest
neighbours, or, in the spirit of Ref. \cite{JD1}, those that are
invariant under any permutation of the parties. While we leave the
first case for further studies, below we focus on the second case
and construct symmetric Bell inequalities with one and two-body
correlators.

By imposing the permutational symmetry, one requires that the
expectation values $\langle\mathcal{M}_{k}^{(i)}\rangle$ and
$\langle\mathcal{M}_k^{(i)}\mathcal{M}_l^{(j)}\rangle$,
with fixed $k,l$ and different $i,j$,
appear in the Bell inequality (\ref{2corrBI}) with the same ``weigths'', i.e.,
$\alpha_i=\alpha$, $\beta_i=\beta$, etc. This means that the
general form of a symmetric Bell inequality with one- and two-body
correlators is
\begin{equation}\label{BellIneq}
I
:=\alpha \mathcal{S}_0+\beta
\mathcal{S}_1+\frac{\gamma}{2}\mathcal{S}_{00}+\delta\mathcal{S}_{01}
+\frac{\varepsilon}{2}\mathcal{S}_{11}\geq -\beta_{C},
\end{equation}
where 
$\alpha,\beta,\gamma,\delta,\varepsilon$ are real parameters. Then, by
$\mathcal{S}_k$ and
$\mathcal{S}_{kl}$ with $k,l=0,1$ we denote the one-
and two-body correlators symmetrized over all observers, i.e.,
\begin{equation}
\mathcal{S}_k=\sum_{i=1}^{N}\langle \mathcal{M}_{k}^{(i)}\rangle, \qquad
\mathcal{S}_{kl}=\sum_{i\neq j=1}^{N}\langle
\mathcal{M}_{k}^{(i)}\mathcal{M}_{l}^{(j)}\rangle.
\end{equation}
%

Geometrically, under this symmetry the polytope
$\mathbbm{P}_2$ is mapped to a simpler one $\mathbbm{P}_{2}^{S}$, which,
idenpendently of $N$, is always five-dimensional and its elements are
vectors
$(\mathcal{S}_0,\mathcal{S}_1,\mathcal{S}_{00},\mathcal{S}_{01},\mathcal{
S}_{11} )$. Accordingly, $\mathbbm{P}_2^{S}$ is fully characterized if one knows
all its facets, which we call \textit{tight symmetric two-body Bell
inequalities}. Moreover, the number of vertices is significantly
reduced from $2^{2N}$ of $\mathbbm{P}_2$ to $2(N^2+1)$ of $\mathbbm{P}_2^S$,
and, as we will see below, vertices of the latter can be conveniently
parameterized by three natural numbers (see appendix A for more details).
Precisely, for a given local deterministic model,
let us denote by $a$, $b$, $c$, and $d$ the amount of
parties whose local expectation values $\langle\mathcal{M}^{(i)}_k\rangle$
$(k=0,1)$ are $\{1,1\}$, $\{1,-1\}$, $\{-1,1\}$,
and $\{-1,-1\}$, respectively. By definition
$a+b+c+d=N$, and therefore all vertices of $\mathbbm{P}_2$ are mapped under the
symmetry to four-tuples $(a,b,c,d)$ forming a tetrahedron $\mathbbm{T}_N$ in
$\mathbbm{N}^3$ whose facets are determined by vanishing one of $a$, $b$, $c$,
or $d$. One can then prove (see appendix A) that all
vertices of $\mathbbm{P}_2^{S}$ are uniquely represented by those
four-tuples that belong to the boundary $\partial \mathbbm{T}_N$ of
$\mathbbm{T}_N$.

Then, for any local deterministic model the one-body
symmetrized expectation values can be expressed within this parametrization as
$\mathcal{S}_k=a+(-1)^k(b-c)-d$ with $k=0,1$. 
Moreover, since for any vertex of $\mathbbm{P}_2$ it holds that
$\mathcal{S}_{kl}=\mathcal{S}_k\mathcal{S}_l-\sum_{i=1}
^N\langle\mathcal{M}_k^{(i)}\rangle\langle\mathcal{M}_l^{(i)}\rangle$
$(k,l=0,1)$, the two-body expectation values are given by
$\mathcal{S}_{ll}=\mathcal{S}_l^2-N$, with $l=0,1$, and
$\mathcal{S}_{01}=\mathcal{S}_0\mathcal{S}_1-(a-b-c+d)$.
As a consequence, computing the classical bound of the
Bell inequality (\ref{BellIneq}) is equivalent to minimizing $I$
being a function of $a,b,c$, and $d$ over the boundary of
$\mathbbm{T}_N$, i.e., $\beta_{C}=-\min_{\partial \mathbbm{T}_N}I$.

\textit{A class of symmetric two-body Bell inequalities.--} Using
the above characterization of the symmetric polytope of two-body
local models, we can now search for particular Bell inequalities
violated by multipartite quantum states. For sufficiently low
number of parties all Bell inequalities corresponding to the
facets of $\mathbbm{P}_2^S$ can be listed with the aid of a
computer algorithm and they will be presented elsewhere
\cite{tech}. Here we present a general class of few-parameter
symmetric tight Bell inequalities and show that they reveal
nonlocality in quantum states for any $N$.

To this end, in Eq. (\ref{BellIneq}) we substitute $\gamma=x^2$
and $\varepsilon=y^2$ with $x,y$ being positive natural numbers,
and $\delta=\sigma xy$, where $\sigma=\pm1$ stands for the sign of
$\delta$. Moreover, let $\alpha_{\pm}=x[\sigma\mu\pm(x+y)]$ with
$\mu\equiv\beta/y$, and assume that $\mu$ is an integer
with opposite parity to $\varepsilon$
($\gamma$) for odd $N$ (even $N$). Exploiting the above parameterization one
then proves that the classical bound of the resulting Bell inequality
is (see appendix B)
\begin{equation}\label{CB}
\beta_C=\frac{1}{2}[N(x+y)^2+(\sigma\mu\pm x)^2-1].
\end{equation}
Interestingly, under the additional assumption that $x$ and $y$ are
coprimes, one can also analyze their tightness by hand
(see Ref. \cite{tech} for more details), and it follows
that this class contains a significant amount of facets of
$\mathbbm{P}_{2}^{S}$ (see Table \ref{table}).
%
Specifically, this class contains those tight Bell inequalities
that are tangent to $\mathbbm{P}_{2}^{S}$ at vertices belonging to
a single facet of $\mathbbm{T}_N$. A particular
example of a Bell inequality of this form, arises
from $x=y=-\sigma=1$, and $\alpha_-=-2$. According to
(\ref{CB}), $\beta_C=2N$ and the resulting Bell inequality is

\begin{equation}\label{desigualdad}
-2\mathcal{S}_0+\frac{1}{2}\mathcal{S}_{00}-\mathcal{S}_{01}+
\frac{1}{2}\mathcal{S}_{11}+2N\geq 0.
\end{equation}
\begin{table}[]
\begin{tabular}{|c|c|c|}
  \hline
$N$    & $\#$ Bell inequalities in the class & Total $\#$ of tight Bell
inequalities\\
  \hline

 5  & 16    & 152\\
 10 & 272   & 2018\\
 15 & 1208  & 7744\\
 20 & 3592  & 21274\\
 \hline
\end{tabular}
\caption{The number of facets (second column) of
$\mathbbm{P}_2^{S}$ that are grasped by our class of Bell
inequalities for various numbers of parties $N$ (first column).
For comparison the third column contains the total number of
facets of $\mathbbm{P}_2^{S}$.}\label{table}
\end{table}

To search for quantum violations of (\ref{desigualdad}), we assume
that all parties have the same pairs of observables, i.e.,
$\mathcal{M}^{(i)}_j=\mathcal{M}_j$ for every $i$.
Without loss of optimality, these observables can be taken to be
equal to $\mathcal{M}_0=\sigma_z$ and
$\mathcal{M}_1=\cos\theta\sigma_z+\sin\theta\sigma_x$ for
$\theta\in[0,\pi]$~\cite{Masanes}. Denoting then by
$\mathcal{B}_N(\theta)$ the Bell operator
constructed from the above observables, the
Bell inequality (\ref{desigualdad}) is violated if there is
$\theta$ such that $\mathcal{B}_N(\theta)\ngeqslant 0$. We have
numerically searched for the lowest negative eigenvalue of
$\mathcal{B}_N(\theta)$ for various values of $N$ and the obtained
results are presented on Fig. \ref{fig:qv}. Clearly, the effective
violation (divided by the classical bound) grows with $N$, and
becomes more robust against misalignments of $\theta$ for large
$N$.
\begin{figure}[t]
\includegraphics[trim=3 2 16 7,clip,width=0.56\columnwidth]{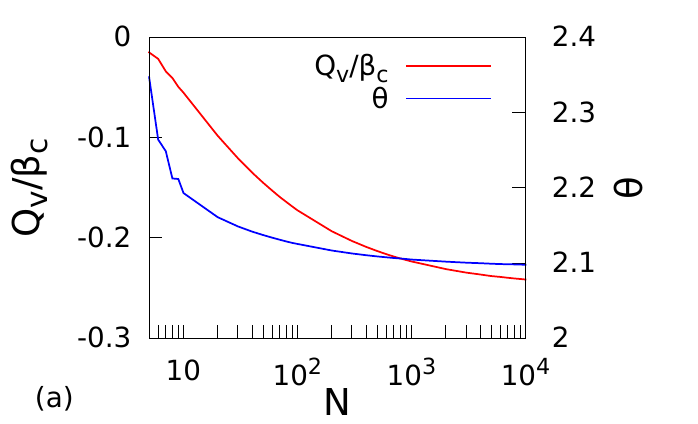}
\includegraphics[trim=3 2.5 11 7,clip,width=0.43\columnwidth,height=3.02cm]{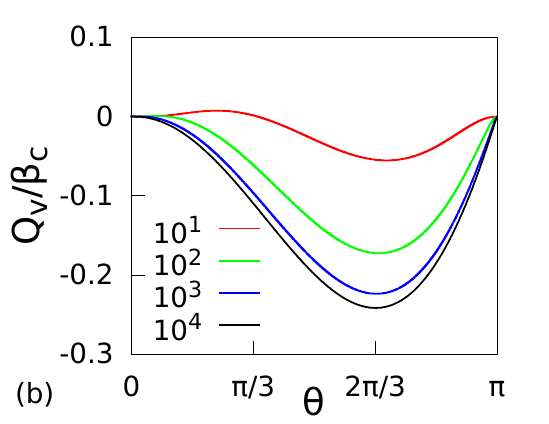}
\caption{(a) The effective (divided by the classical bound) maximal violation of
Ineq. (\ref{desigualdad}) (red line) and the corresponding angle
$\theta$ in $\mathcal{M}_1$ (blue line) as functions of $N$.
(b) Effective violation of Ineq. (\ref{desigualdad}) as a function of
$\theta$ for $N=10^k$ with $k=1,2,3,4$. For large $N$ the
violation is robust against misalignments of the second
observable.}\label{fig:qv}
\end{figure}
Also, the corresponding eigenstate of $\mathcal{B}_N(\theta)$,
i.e., the state maximally violating (\ref{desigualdad}) is always symmetric
(but there are also antisymmetric states violating this inequality),
that is, it is invariant under permutation of any two particles.
Consequently, since any (also mixed) $N$-qubit symmetric state is entangled if,
and only if, it is genuinely multipartite entangled (see e.g. Refs.
\cite{Symm}), our Bell inequalities detect states that
have genuine multipartite entanglement.

{\it Nonlocality of physically relevant states.--}As we have just demonstrated,
the two-body multipartite Bell inequalities
detect nonlocality of multipartite quantum states. But are
these inequalities powerful enough to reveal nonlocality in ``physically
relevant'' states, as for instance ground states of spin models that naturally
appear in many-body physics? Here we show that this is the case by constructing
a class of two-body symmetric Bell inequalities that are violated by the Dicke
states \cite{Dicke}. These are $N$-qubit states spanning the $(N+1)$-dimensional
symmetric subspace of $(\mathbbm{C}^2)^{\ot N}$ and read
\begin{equation}
 \ket{D_N^{k}}=\mathcal{S}(\ket{\{0,N-k\},\{1,k\}}) \quad (k=0,\ldots,N),
\end{equation}
where $\ket{\{0,N-k\},\{1,k\}}$ is any pure product vector with $N-k$ qubits in
the state $\ket{0}$ and $k$ in the state $\ket{1}$, while $\mathcal{S}$ denotes
symmetrization over all parties. It is worth mentioning that
$\ket{D_N^k}$ are genuinely multipartite entangled for any $k\neq 0,N$.
Moreover, their entanglement properties have been extensively studied in
the literature (see e.g. Refs. \cite{LOR,DickeOtfried} and references therein),
and the state $\ket{D_6^3}$ was recently generated experimentally
\cite{DickeExp}.

In many-body physics, the Dicke states arise naturally as the lowest-energy
eigenstates of the isotropic Lipkin-Meshkov-Glick Hamiltonian \cite{LMG}:
\begin{equation}
H=-\frac{\lambda}{N}\sum_{\substack{i,j=1\\i<j}}^{N}\left(\sigma_x^{(i)}
\sigma_x^ { (j) } + 
\sigma_y^{(i)}\sigma_y^{(j)}\right)-h\sum_{i=1}^{N}\sigma_z^{(i)},
 \label{eq:LMG}
\end{equation}
which describes $N$ spins interacting
through the two-body ferromagnetic
coupling $(\lambda>0)$, embedded into the magnetic
field acting along the $z$ direction of strength $h\geq 0$. Again,
$\sigma_{a}^{(i)}$ ($a=x,y,z$) are the Pauli matrices acting at site $i$.


In what follows we consider the case of weak magnetic field
applied to the system, precisely $h \leq \lambda /N$. Then,
the ground state of $H$ is $\ket{D^{N/2}}$ for even $N$ and $\ket{D^{\lceil
N/2\rceil}}$ for odd $N$, except for the case of $h=0$ and odd $N$, for which the lowest energy is two-fold degenerate and the corresponding subspace is spanned by $\ket{D_{N}^{k}}$, with $k=\lfloor N/2\rfloor$ and $k=\lceil N/2\rceil$.

The class of tight two-body symmetric Bell inequalities
that we use to detect nonlocality of the above Dicke states is obtained by
taking $\alpha_N=N(N-1)(\lceil N/2\rceil-N/2)$, $\beta_N=\alpha_N/N$, $\gamma_N=N(N-1)/2$, $\delta_N=N/2$, and
$\varepsilon_N=-1$ in Eq.
(\ref{BellIneq}). This choice of parameters allows us to compute
analytically the classical bounds of the resulting Bell
inequalities for any $N$. Precisely, the minimization over
$\partial\mathbbm{T}_N$
gives
\begin{equation}
\beta_C(N)=\frac{1}{2}N(N-1)\left\lceil\frac{N+2}{2}\right\rceil
%
\end{equation}
and also allows one to find five vertices at which these Bell inequalities
are tangent to $\mathbbm{P}_2^S$, ensuring their tightness.
(cf. appendix C for the proof). It should be noticed that these Bell
inequalities are independent of the class presented previously,
and for $N=2$ they reproduce the CHSH Bell inequality \cite{CHSH}.

To prove that the resulting Bell inequalities are indeed violated
by the Dicke states, let again $\mathcal{M}_j^{(i)}$ $(j=0,1)$ be
the qubit dichotomic observables at site $i$. Denoting by
$\mathcal{B}_N$ the resulting Bell operator, the direct way to
reach the goal is to minimize the mean value
$\langle D_N^k|\mathcal{B}_N|D_N^k\rangle$ with $k=\lfloor
N/2\rfloor,\lceil N/2\rceil$. As before, we assume for simplicity that all observers measure the same pair of observables, i.e.,
$\mathcal{M}^{(i)}_j=\mathcal{M}_j$. This makes $\mathcal{B}_N$
permutationally invariant, which together with the fact that the
Dicke states are symmetric significantly simplifies the above
problem. In fact, it follows that
%
$\langle
D_N^k|\mathcal{B}_N|D_N^k\rangle=\Tr(\rho_{N}^k\widetilde{\mathcal{
B}}_N)$,
%
where $\widetilde{\mathcal{B}}_N$ stands for the
two-qubit ``reduced'' Bell operator
\begin{eqnarray}\label{BN}
\widetilde{\mathcal{B}}_N&=&\beta_{C}(N)\mathbbm{1}
_4+\tfrac{N}{2}\alpha_N(\mathcal { M
}_0\ot\mathbbm{1}_2+\mathbbm{1}_2\ot\mathcal{M}_0)\nonumber\\
&&+\tfrac{N(N-1)}{2}\left[\gamma_N\mathcal{M}_0\ot\mathcal{M}_0+
\varepsilon_N\mathcal{M}_1\ot\mathcal{M}_1\right.\nonumber\\
&&\left.\hspace{2cm}+\delta_N\left(\mathcal{M}_0\ot\mathcal{M}_1+\mathcal{M}
_1\ot\mathcal{M}_0\right)\right],\nonumber\\
&&+\tfrac{N}{2}\beta_N(\mathcal{M}_1\ot\mathbbm{1}_2+\mathbbm{1}
_2\ot\mathcal { M}_1)
\end{eqnarray}
with $\mathbbm{1}_d$ being a $d\times d$ identity matrix, and
$\rho_{N}^k$ denotes any two-qubit subsystem of $\ket{D_{N}^k}$.
Notice that the latter can be computed by hand for any $N$ and $k$
(see appendix C).

As before, let us finally set the measurements to
$\mathcal{M}_0=\sigma_z$ and
$\mathcal{M}_1=\cos\theta\sigma_z+\sin\theta\sigma_x$ with
$\theta\in[0,\pi]$, and choose the particular Dicke state with
$k=\lceil N/2\rceil$ excitations, for any $N$.
Then, $\langle D_N^k|\mathcal{B}_N|D_N^k\rangle$ can be computed
to be $4\lfloor N/2\rfloor\sin^2(\theta/2)[(\lceil
N/2\rceil+1)\sin^2(\theta/2)-1]$, and the latter attains its
minimum for
\begin{equation}
\theta_{\min}^N=\pm \arccos\left(\frac{\lceil N/2\rceil}{\lceil
N/2\rceil+1}\right),
%
\end{equation}
resulting  in the following quantum violations
\begin{equation}
\langle
D_N^{\lceil
N/2\rceil}|\mathcal{B}_N|D_N^{\lceil N/2\rceil}\rangle=
-\frac{\lfloor N/2\rfloor}{\lceil N/2\rceil+1}.
\end{equation}
Similar values can be obtained for $k=\lfloor N/2\rfloor$:
it is enough to rename (flip) the
outcomes of $\mathcal{M}_j$ $(j=0,1)$, because
$\ket{D_N^{\lfloor N/2\rfloor}}$ is
obtained from $\ket{D_N^{\lceil N/2\rceil}}$ by swapping the
elements of the computational basis $\{\ket{0},\ket{1}\}$.


To summarize, our Bell inequalities are
violated by the Dicke states for any $N$,
although the effective violation decays with
$N$ as $1/N^3$ (see Fig. \ref{fig2}). It should be stressed that,
even though from the previous analysis one may conclude that this
violation is purely bipartite, this is certainly not the case. 
The Dicke states are symmetric, and therefore any marginal bipartite correlations obtained from them in a Bell experiment with the same two dichotomic observables per site are local; otherwise all bipartite marginal correlations would be nonlocal, contradicting the fact that in this case quantum correlations are monogamous \cite{TV}. This also means that our results provide further examples, after \cite{Lars}, of local marginal bipartite correlations that are only compatible with global multipartite nonlocal correlations (see also Ref. \cite{tech}).


%
\begin{figure}[t]
\includegraphics[width=0.49\columnwidth]{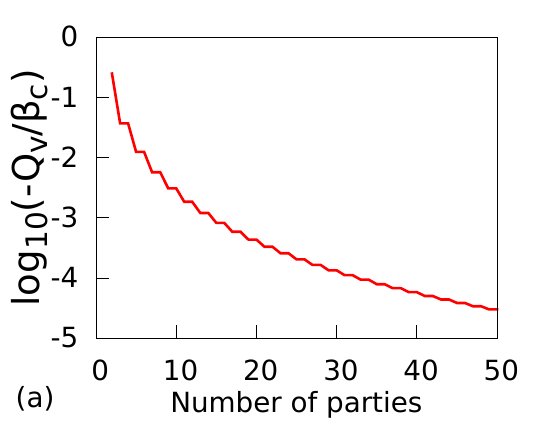}
\includegraphics[width=0.49\columnwidth]{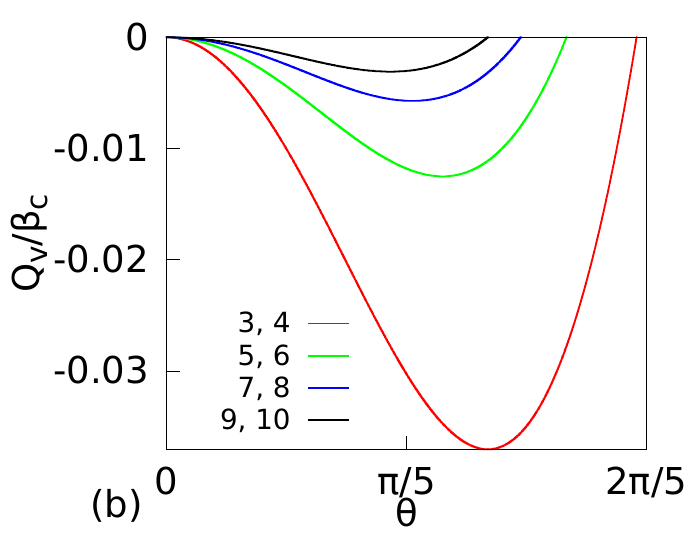}
\caption{(a) Effective violation of Ineq. (\ref{desigualdad}) by
the Dicke states $\ket{D_N^{\lceil N/2\rceil}}$ as a function of
$N$. The violation decays with $N$ as $1/N^3$. (b) Effective violation as a
function of $\theta$ for various values of $N$. }\label{fig2}
\end{figure}
%

\textit{Discussion and conclusion.--}Bell inequalities allow
detecting nonlocality in composite quantum states. They involve
expectation values of products of local measurements performed by
the observers. A natural question is about the minimal amount of
knowledge (in terms of the size of correlators) that is needed to
reveal nonlocality. Here we have demonstrated that multipartite two-body Bell
inequalities are enough to witness nonlocality for an arbitrary number of parties.

Interestingly, our inequalities are violated by
symmetric states, making our result feasible from the experimental
point of view. First, such states appear naturally as ground
states of models that can be realized with ultracold atoms or
ions, such as Lipkin-Meshkov-Glick like models with long range
interactions (for ionic spin 1/2 and spin 1 realizations see
\cite{Porras,Grass}, for cold atoms in nanophotonic waveguides see
\cite{Chang}), or degenerated ground states of the ferromagnetic Heisenberg
model \cite{Sachdev}. Second, in multipartite symmetric states one-body and
two-body expectation values can be addressed via collective measurements of
total spin operators $S_{\alpha}=(1/2)\sum_{i=1}^{N}\sigma_{\alpha}^{(i)}$
$(\alpha=x,y,z)$ and simple second-order functions thereof, respectively (see
e.g. Ref. \cite{Molmer}). We have shown, for instance, that the nonlocality of
some of the half-filled Dicke states can be certified from only two correlators
$\langle \sigma_z^{(1)}\sigma_z^{(2)}\rangle$ and $\langle
\sigma_z^{(1)}\sigma_x^{(2)}\rangle$, which within the above framework
can be expressed as $\langle \sigma_z^{(1)}\sigma_z^{(2)}\rangle=
(4\langle S_z^{2}\rangle-N)/N(N-1)$ and $\langle
\sigma_z^{(1)}\sigma_x^{(2)}\rangle=([S_z,S_x]_{+}\rangle)/N(N-1)$
with $[\cdot]_{+}$ denoting the anticommutator. Such measurements are
routinely realized in atomic systems with current
experimental technologies, such as spin polarization spectroscopy
\cite{Hammerer,Isart}.

{\it Acknowledgments.--}Discussions with J. Stasi\'nska are
greatly acknowledged. This work is supported by Spanish DIQIP
CHIST-ERA, FIS2010-14830 projects and AP2009-1174 FPU PhD grant, EU IP
SIQS, ERC AdG QUAGATUA and StG PERCENT. R. A. also acknowledges
the Spanish MINECO for the Juan de la Cierva scholarship.


\appendix

\section{APPENDIX}


\subsection{Appendix A: Characterization of vertices of $\mathbbm{P}_2^S$}

Let $\mathbbm{P}_2$ be the polytope of all local models constructed from
$\mathbbm{P}$ by forgetting the correlators of order higher than two, and
let $V$ be a set containing $2^N$ vertices of $\mathbbm{P}_2$. Recall
that the
latter are those local strategies for which
$\langle\mathcal{M}^{(i)}_k\mathcal{M}^{(j)}_l\rangle=\langle\mathcal{M}^{(i)}_k
\rangle\cdot \langle\mathcal{M}^{(j)}_l\rangle$ and
$\langle\mathcal{M}^{(i)}_k\rangle=\pm1$, for $i,j=1,\ldots,N$
and $k,l=0,1$. Moreover, by $\mathbbm{P}_{2}^S$ we denote the
image of $\mathbbm{P}_2$ under symmetrization, i.e., the set of five-dimensional
vectors
\begin{equation}\label{five}
(\mathcal{S}_0,\mathcal{S}_1,\mathcal{S}_{00},\mathcal{S}_{01},\mathcal{S}_{11})
\end{equation}
with $\mathcal{S}_{k}$ and $\mathcal{S}_{kl}$ $(k,l=0,1)$ defined by
\begin{equation}
\mathcal{S}_k=\sum_{i=1}^{N}\langle\mathcal{M}^{(i)}_k\rangle,\qquad
\mathcal{S}_{kl}=\sum_{\substack{i,j=1\\i\neq
j}}^{N}\langle\mathcal{M}^{(i)}_k
\mathcal{M}_l^{(j)}\rangle,
\end{equation}
that are computed for all elements of $\mathbbm{P}_2$. Analogously, by $V_S$ we
denote the set of vertices of $\mathbbm{P}_{2}^S$. Our aim now is to
characterize the elements of $V_S$ and identify those vertices
of $\mathbbm{P}_2$ that are mapped onto the vertices of $\mathbbm{P}_{2,S}$. In
particular, we will demonstrate that $|V_S|=2(N^2+1)$, which
is a significantly smaller number than $|V|=2^{2N}$.

To this end, for every element of $V$ we denote by
\begin{equation}
x_i= \langle\mathcal{M}_0^{(i)}\rangle,\qquad
y_i=\langle\mathcal{M}_1^{(i)}\rangle,
\end{equation}
the pair of local deterministic expectation values (those that assume values
$\pm 1$) for party $i$, and by $\{x_i,y_i\}$ the corresponding local strategy.
Then, we notice that the values of $\mathcal{S}_0$ and $\mathcal{S}_1$ do not
depend on particular local strategies applied by the parties but rather on their
amount. This suggests introducing the following parametrization:
\begin{eqnarray}\label{parametrization}
a&=&
\#\{i\in\{1,\ldots,N\}\,|\,x_i=1,y_i=1\}
\nonumber\\
b&=&\#\{i\in\{1,\ldots,N\}\,|\,x_i=1,y_i=-1\},\nonumber\\
c&=&\#\{i\in\{1,\ldots,N\}\,|\,x_i=-1,y_i=1\},\nonumber\\
d&=&\#\{i\in\{1,\ldots,N\}\,|\,x_i=-1,y_i=-1\}.
\end{eqnarray}
In other words, for a given vertex of $\mathbbm{P}_2$,
$a$, $b$, $c$, and $d$ stand for the number of
parties who apply one of the four different local strategies $\{1,1\}$,
$\{1,-1\}$, $\{-1,1\}$, or $\{-1,-1\}$, respectively.
Clearly, $a+b+c+d=N$, and by means
of these four numbers, the symmetrized local expectation values
$\mathcal{S}_k$ $(k=0,1)$ can be expressed as
\begin{equation}\label{S0S1}
\mathcal{S}_0=a+b-c-d, \qquad
\mathcal{S}_1=a-b+c-d.
\end{equation}
Furthermore, since for every element of $V$ it holds that
\begin{equation}
\mathcal{S}_{xy}=\mathcal{S}_x\mathcal{S}_y-\sum_{i=1}^{N}\langle\mathcal{M}_x^{
(i) } \rangle\langle\mathcal { M } _y^{(i)}\rangle\qquad (x,y=0,1),
\end{equation}
by using the parametrization (\ref{parametrization}) together with Eqs.
(\ref{S0S1}), we can rewrite the two-body symmetrized
expectation values as
\begin{eqnarray}\label{jazz}
\mathcal{S}_{00}&=&\mathcal{S}_0^2-N=(a+b-c-d)^2-N, \nonumber\\
\mathcal{S}_{11}&=&\mathcal{S}_1^2-N=(a-b+c-d)^2-N,
\end{eqnarray}
and
\begin{eqnarray}
\mathcal{S}_{01}&=&\mathcal{S}_0\mathcal{S}_1-
\sum_{i=1}^{N}\langle\mathcal{M}_x^{
(i) } \rangle\langle\mathcal { M } _y^{(i)}\rangle\\
&=&(a+b-c-d)(a-b+c-d)-(a-b-c+d).\nonumber
\end{eqnarray}
Thus, all vertices of $\mathbbm{P}_2$ are mapped under
symmetrization onto elements of $\mathbbm{P}_{2}^S$ that
can later be parameterized by elements of the following set
%
$
\{(a,b,c,d)\in\mathbbm{N}^4\,|\,a+b+c+d=N\},$
%
which is isomorphic to a tetrahedron in $\mathbbm{N}^3$
\begin{equation}
\mathbbm{T}_N=\{(a,b,c)\in\mathbbm{N}^3\,|\,a+b+c\leq N\}.
\end{equation}
In addition, the facets of $\mathbbm{T}_N$ contain those three-tuples
$(a,b,c)$ for which either $abc=0$ or $a+b+c=N$ (equivalently $d=0$).
The cardinality of $\mathbbm{T}_N$, which is basically the number of all
possible choices of four natural numbers
summing up to $N$, and it amounts to $(1/6)(N+1)(N+2)(N+3)$.

In what follows, we show that vertices of $\mathbbm{P}_{2}^S$ are
uniqely represented by all those $4$-tuples from $\mathbbm{T}_N$ that belong to
its boundary $\partial\mathbbm{T}_N$, i.e., those for which the condition
$abcd=0$ is satisfied.

\begin{thm}\label{thm:app}
Let us denote by $\varphi:\mathbbm{T}_N\mapsto \mathbbm{P}_{2}^S$
the above parametrization, i.e.,
\begin{equation}\label{parametr}
\varphi((a,b,c,d))=(\mathcal{S}_0,\mathcal{S}_1,\mathcal{S}_{00},\mathcal{S}_{
01},\mathcal{S}_{11}).
\end{equation}
Then $\varphi(p)$ is a vertex of $\mathbbm{P}_{2}^S$ iff
$p\in\partial\mathbbm{T}_N$.
\end{thm}
\begin{proof}We start from the ``only if'' part.
Assume on the contrary that $p=(a,b,c,d)$ belongs to the interior of
$\mathbbm{T}_N$, which means that all its components are larger than zero,
(i.e., $a,b,c,d\geq 1$). Then, let us consider a vector
$v=(1,-1,-1,1)\notin\mathbbm{T}_N$ and notice that the values of
$\mathcal{S}_k$ and $\mathcal{S}_{kk}$ with $k=0,1$ are constant along the line
$p+\lambda v$ for any $\lambda\in \mathbbm{R}$, while
$\mathcal{S}_{01}(p+\lambda v)=\mathcal{S}_{01}(p)-4\lambda$. Hence, for any
$\alpha,\beta>0$,
\begin{eqnarray}
&&\hspace{-0.5cm}\alpha\varphi(p+\beta v)+\beta\varphi(p-\alpha v)\nonumber\\
&&=\alpha(\mathcal{S}_0(p),\mathcal{S}_1(p),\mathcal{S}_{00}(p),\mathcal{S}_{01}
(p)-4\beta,\mathcal{S}_{11}(p))\nonumber\\
&&\hspace{0.3cm}+\beta(\mathcal{S}_0(p),\mathcal{S}_1(p),\mathcal{S}_{00}(p),
\mathcal { S } _ { 01 }
(p)+4\alpha,\mathcal{S}_{11}(p))\nonumber\\
&&=(\alpha+\beta)(\mathcal{S}_0(p),\mathcal{S}_1(p),\mathcal{S}_{00}(p),
\mathcal{S}_{01}(p),\mathcal{S}_{11}(p))\nonumber\\
&&=(\alpha+\beta)\varphi(p),
\end{eqnarray}
which allows us to express $\varphi(p)$ as
\begin{equation}
\varphi(p)=\frac{\alpha}{\alpha+\beta}\varphi(p+\beta
v)+\frac{\beta}{\alpha+\beta}\varphi(p-\alpha v).
\end{equation}
Now choose $\alpha=\min\{a,d\}$ and $\beta=\min\{b,c\}$.
Then, both $p+\beta v$ and $p-\alpha v$ belong to the boundary
of $\mathbbm{T}_N$, and consequently $\varphi(p)\in\mathbbm{P}_{2}^S$
represented by $p$ can be written as a convex combination of
two other elements $p+\min\{a,d\}v$ and $p-\min\{b,c\}v$
of $\mathbbm{P}_{2}^S$. This implies that $\varphi(p)$ cannot
be extremal.

In order to prove the ``if'' part, assume that
$p\in\partial\mathbbm{T}_N$ and
$\varphi(p)=(\mathcal{S}_0,\mathcal{S}_1,\mathcal{S}_{00},\mathcal{S}_{01},
\mathcal{S}_{11})$ is not a vertex of
$\mathbbm{P}_{2}^S$. Then, $\varphi(p)$ can be decomposed into a convex
combination of vertices of $\mathbbm{P}_{2}^S$ that
are represented within our parametrization by
$p_i=(a_i,b_i,c_i,d_i)\in\mathbbm{T}_N$, i.e.,
\begin{equation}\label{Olot}
\varphi(p)=\sum_{i=0}^{k}\lambda_i\varphi(p_i)
\end{equation}
with $0<\lambda_i<1$ summing up to unity, and
\begin{equation}\label{combination}
\varphi(p_i)=(\mathcal{S}_0^{(i)},\mathcal{S}_1^{(i)},\mathcal{S}_{00}^{(i)},
\mathcal {S}_{ 01 }^{(i)},\mathcal{S}_{11}^{(i)}).
\end{equation}
By combining Eqs. (\ref{parametr}) and (\ref{combination}), Eq. (\ref{Olot})
is equivalent to the following five equations:
%
%
\begin{equation}\label{LLobregat}
\mathcal{S}_l=\sum_{i=0}^k\lambda_i\mathcal{S}_l^{(i)},\qquad
\mathcal{S}_{ll}=\sum_{i=0}^k\lambda_i\mathcal{S}_{ll}^{(i)}
\end{equation}
for $l=0,1$, and
\begin{equation}
\mathcal{S}_{01}=\sum_{i=0}^k\lambda_i\mathcal{S}_{01}^{(i)}.
\end{equation}
Since for all vertices of $\mathbbm{P}_2$ it holds that
$\mathcal{S}_{ll}^{(i)}=[\mathcal{S}_{l}^{(i)}]^2-N$ [cf. Eqs. (\ref{jazz})],
Eqs. (\ref{LLobregat}) imply that $\mathcal{S}_l^{(i)}$ must satisfy
\begin{equation}\label{eq:cond30}
\sum_{i}\lambda_i\left(\mathcal{S}_l^{(i)}\right)^2=\left(\sum_{i}\lambda_i
\mathcal{S}_l^{(i)}\right)^2 \qquad (l=0,1).
\end{equation}
If we think of Eq. (\ref{eq:cond30}) as a quadratic equation
for a particular $\mathcal{S}_{l}^{(m)}$, i.e., 
i.e.,
\begin{eqnarray}
\lambda_{m}(\lambda_{m}-1)\left(\mathcal{S}_l^{(m)}\right)^2+
2\lambda_m\mathcal{S}_l^{(m)}\sum_{i\neq m}\lambda_i\mathcal{S}_l^{(i)}
\nonumber\\
+\left(\sum_{i\neq m}\lambda_i\mathcal{S}_{l}^{(i)}\right)^2-\sum_{i\neq
m}\lambda_i\left(\mathcal{S}_l^{(i)}\right)^2=0
\end{eqnarray}
it has real solutions if and only if its discriminant is nonnegative,
which in turn holds iff
\begin{equation}\label{condition1}
-4\lambda_0\sum_{\substack{i<j\\i,j\neq
m}}\lambda_i\lambda_j\left(\mathcal{S}_l^{(i)}-\mathcal{S}_{l}^{(j)}\right)
^2\geq 0.
\end{equation}
Since all $\lambda$'s are positive, the above condition is fulfilled iff
$\mathcal{S}_l^{(i)}=\mathcal{S}_l^{(j)}$ for all $i,j\neq m$ and $l=0,1$.
Due to the fact that (\ref{condition1}) must be obeyed
for any $m$, we have eventually that
\begin{equation}\label{Gdansk}
\mathcal{S}_l^{(i)}=\mathcal{S}_l^{(j)}=\mathcal{S}_l
\end{equation}
for any $i,j=1,\ldots,k$ and $l=0,1$.

On the other hand, the assumption that $\varphi(p)$ is not a vertex of
$\mathbbm{P}_{2}^S$, i.e., that it can be decomposed as in
(\ref{Olot}), means that $\mathcal{S}_{01}^{(i)}$ cannot be equal, as otherwise
$p_i$ are all the same. If we then express
$\mathcal{S}_{01}=\mathcal{S}_0\mathcal{S}_1-(a-b-c+d)$
and $\mathcal{S}_{01}^{(i)}=\mathcal{S}_{0}^{(i)}\mathcal{S}_{1}^{(i)}
-(a_i-b_i-c_i+d_i)$, this, in virtue of Eq. (\ref{Gdansk}),
implies
\begin{equation}
a-b-c+d=\sum_{i}\lambda_i(a_i-b_i-c_i+d_i).
\end{equation}
If we further note that 
\begin{equation}\label{Danzig}
a_i+b_i+c_i+d_i=N=a+b+c+d
\end{equation}
must hold for any $i$, we infer
that $a$, $b$, $c$, and $d$
are convex combinations of
$a_i$, $b_i$, $c_i$, and $d_i$, respectively, and, as a result
\begin{equation}
 p=\sum_{i=1}^{k}\lambda_ip_{i}.
\end{equation}
In order to complete the proof (that is, to reach the contradiction with the
assumption) it is enough to notice that $p\in\mathrm{int}\mathbbm{T}_N$, since
not all of $p_i$ can belong to the same facet of $\mathbbm{T}_{N}$. In
fact, if all $p_i$ belong to the same facet of the tetrahedron, one of their
coordinates (the same one for all $i$) must be zero (for instance, $a_i=0$).
Then, it directly follows from Eqs. (\ref{Gdansk}) and (\ref{Danzig}) that
all $p_i$s are equal, contradicting the assumption
that (\ref{Olot}) is a proper convex combination. Consequently, $p$ belongs to
the interior of the tetrahedron, which contradicts the assumption that
$p\in\partial\mathbbm{T}_N$, completing the proof.
\end{proof}

\subsection{Appendix B: A class of symmetric Bell inequalities}
\label{AppB}

In this appendix, we revisit the three-parameter
class of symmetric multipartite Bell inequalities and compute in detail its
classical bound. Recall for this purpose that $\gamma=x^2$ and
$\varepsilon=y^2$, where $x$ and $y$ are
positive integers. Assume that
$\mu=\beta/y\in\mathbbm{Z}$, $\delta=\sigma xy$ and
$\alpha_{\pm}=x[\sigma\mu\pm(x+y)]$ with $\sigma$ denoting the sign of $\delta$,
and further that the parity of $\mu$ is
opposite to that of $\varepsilon$ ($\gamma$) for odd
$N$ (even $N$). In what follows we will show that the classical bound of the
resulting Bell inequality is
\begin{equation}\label{StarWars}
\beta_{C}=\frac{1}{2}\left[N(x+y
)^2+(\sigma\mu\pm x)^2\right ]-\frac{1}{2}.
\end{equation}
First, note that for all local deterministic models,
the left-hand side of (\ref{BellIneq}) can be rewritten as
\begin{equation}
I=\alpha\mathcal{S}_0+\beta\mathcal{S}_1+\frac{\gamma}{2}(\mathcal{S}
_0^2-N)+\delta(\mathcal{S}_0\mathcal{S}_1-z)+\frac{\varepsilon}{2}(\mathcal{S}
_1^2-N),
\end{equation}
where $z=a-b-c+d$. The above choice of parameters
further implies that
\begin{eqnarray}
\nonumber\\
I&=&\frac{x^2}{2}\mathcal{S}_0^2+\sigma
xy\mathcal{S}_0\mathcal{S}_1+\frac{y^2}{2}
\mathcal{S}_1^2-\frac{N}{2}(x+y)-\sigma xyz\nonumber\\
&&+x[ \sigma\mu\pm(x+y) ] \mathcal
{ S } _0+\beta\mathcal { S } _1 \,,
\end{eqnarray}
which can rewritten as

%
%
\begin{eqnarray}
I&=&\frac{1}{2}\left(x\mathcal{S}_0+\sigma y\mathcal{S}_1+\sigma\mu\pm
x\right)^2-\frac{1}{2}(\sigma\mu\pm x)^2\nonumber\\
&&+xy(\pm\mathcal{S}_0\mp\sigma
\mathcal{S}_1-\sigma z)-\frac{1}{2}N(x+y).
\end{eqnarray}
Within the parameterization (\ref{S0S1}),
$\pm\mathcal{S}_0\mp\sigma \mathcal{S}_1-\sigma z =4r-N$,
where $r$ depends on $\alpha$ and the sign of $\delta$ (i.e., $\sigma$)
as follows:
%
\begin{equation}
r=\left\{
\begin{array}{ll}
b, \quad& \mathrm{for}\qquad\alpha_{+},\sigma=1\\
a, \quad& \mathrm{for}\qquad\alpha_{+},\sigma=-1\\
c, \quad& \mathrm{for}\qquad\alpha_{-},\sigma=1\\
d, \quad& \mathrm{for}\qquad\alpha_{-},\sigma=-1\\
\end{array}
\right..
\end{equation}
As a consequence
\begin{eqnarray}
I&=&\frac{1}{2}\left(x\mathcal{S}_0+\sigma y\mathcal{S}_1+\sigma\mu\pm
x\right)^2+4xyr\nonumber\\
&&-\frac{1}{2}\left[(\sigma\mu\pm x)^2+N(x+y)^2\right].
\end{eqnarray}
When comparing the above expression with Eq.
(\ref{StarWars}), it follows that in
order to prove the latter to be the classical bound
of $I$, it is enough to show that
\begin{equation}\label{CanalOlimpic}
\left(x\mathcal{S}_0+\sigma y\mathcal{S}_1+\sigma\mu\pm
x\right)^2+8xyr\geq 1.
\end{equation}
Since both $x$ and $y$ are positive integers,
the above inequality is trivially satisfied if $r\neq 0$.
For $r=0$ (i.e., when optimizing over the facets of
$\mathbbm{T}_N$), observe that the expression in the parentheses
is integer. Therefore, the inequality
(\ref{CanalOlimpic}) is not satisfied only if the parenthesis is equal to
zero. To prove that this may not happen,
we will show that the above
assumptions guarantee that the parity of $x\mathcal{S}_0+\sigma
y\mathcal{S}_1+\sigma\mu\pm x$ is always odd. Let us consider the cases of odd
and even $N$ separately. For odd $N$, one finds that both $\mathcal{S}_0$ and
$\mathcal{S}_1$ are odd. Hence, $x\mathcal{S}_0+\sigma
y\mathcal{S}_1+\sigma\mu\pm x$ has the same parity as $y+\mu$. Then, by
assumption, $\mu$ has opposite parity to $\varepsilon$. Noting
that $\varepsilon=y^2$, and that
both $y^2$ and $y$ have the same parity, we conclude that $y + \mu$
is odd, meaning that the above expression cannot be zero. For even $N$,
$\mathcal{S}_0$ as well as $\mathcal{S}_1$ are even, implying that
$x\mathcal{S}_0+y\sigma\mathcal{S}_1$ is even. The assumptions further
guarantee that $\sigma \mu \pm x$ is odd,
and therefore the expression is nonzero.

Under the additional assumption that $x$ and $y$ are coprimes one is also
able to determine analitycally all vertices of $\mathbbm{P}_2^S$ saturating
these Bell inequalities, and thus check their tightness. The detailed
considerations will be presented elsewhere \cite{tech}.


\section{Appendix C: Violation of two-body Bell inequalities by the Dicke
states}
\label{AppC}

Here we compute in detail the classical bound of the
Bell inequality violated by the Dicke states, and present the explicit form
of the reduced bipartite subsystem of the Dicke
state $\ket{D_N^{\lceil N/2\rceil}}$.

\subsection{The classical bound}

The explicit form
of these Bell inequalities for even $N$ reads
\begin{equation}\label{BellEven}
I^e_N=\frac{N(N-1)}{4}\mathcal{S}_{00}+\frac{N}{
2 } \mathcal { S } _ { 01 } -\frac{1}{2}\mathcal { S } _ { 11 } \geq -\beta_C,
\end{equation}
while for odd $N$,
\begin{eqnarray}\label{BellOdd}
I^o_{N}&=&
\frac{1}{2}\binom{N} { 2 }
\mathcal{S}_{00}+\frac{N}{2}\mathcal{S}_{01}-\frac{1}
{ 2 } \mathcal { S } _ {
11 }\nonumber\\
&&
+\frac{N(N-1)}{2}\mathcal{S}_0+\frac{N-1
} {2}\mathcal{S}_1
 \geq -\beta_C.
\end{eqnarray}

Our aim is to prove that the minimal value of
the Bell inequalities
(\ref{BellEven}) and (\ref{BellOdd})
over local models is given by


%
%
\begin{equation}\label{betaC}
\min_{\partial \mathbbm{T}_N}I_N^{e/o}=-\frac{1}{4}\left\{
\begin{array}{cc}
N(N-1)(N+2),& \quad N\, \mathrm{even}\\[1ex]
N(N-1)(N+3),& \quad N\, \mathrm{odd},
\end{array}
\right.
%
\end{equation}
and, their classical bound is then
$\beta_{C}=-\min_{\partial\mathbbm{T}_{N}}I_N^{e/o}$.

For this purpose, let us first notice that for
classical correlations the following contraints hold:
\begin{equation}\label{ineqs}
-N\leq \mathcal{S}_{00},\mathcal{S}_{11}\leq N(N-1),
\quad |\mathcal{S}_{01}|\leq N(N-1),
\end{equation}
and $|\mathcal{S}_k|\leq N$ with $k=0,1$. Hence, the first term in
Eqs. (\ref{BellEven}) and (\ref{BellOdd}) is the dominant one (it is of
the fourth order in $N$, while the remaining ones are of the second or third
orders in $N$).
This means that in order to minize $I_N^{e/o}$
over the local models, one needs to make the term
containing $\mathcal{S}_{00}$ small. Since
$\mathcal{S}_{00}=\mathcal{S}^{2}_0-N$,
the above expression suggests treating $\mathcal{S}_0$ as a parameter
with which to lower the number of variables in the
optimization. Then, among all the solutions parameterized by $\mathcal{S}_0$ we can choose the smallest one.

Let us now switch to the parametrization in terms of
(\ref{parametrization}). We already know that
$\mathcal{S}_0=a+b-c-d$ which together with
$a+b+c+d=N$, allows us to decrease the number of free variables in the
optimization down to two.
That is, by taking, for instance,
\begin{equation}\label{param}
a=\frac{1}{2}(N+\mathcal{S}_0)-b,\qquad c=\frac{1}{2}(N-\mathcal{S}_0)-d,
\end{equation}
we can express the two-body expectation values as
\begin{equation}\label{Penedes}
\mathcal{S}_{00}=\mathcal{S}_0^2-N, \qquad \mathcal{S}_{11}=[N-2(b+d)]^2-N,
\end{equation}
and
\begin{equation}\label{Besalu}
\mathcal{S}_{01}={\mathcal S}_0[N-1-2(b+d)]+2(b-d),
\end{equation}
where we now consider $b$ and $d$ as free variables that are nonnegative
integers constrained as
\begin{equation}\label{Poblet}
0\leq b\leq \frac{1}{2}(N+\mathcal{S}_0),\qquad 0\leq d\leq
\frac{1}{2}(N-\mathcal{S}_0).
\end{equation}

Notice finally that from theorem
\ref{thm:app} it follows that, in order to find $\beta_C$, it suffices to
minimize $I_N^{e/o}$ over the $4$-tuples $(a,b,c,d)$ belonging to the boundary
of the tetrahendron, i.e., those for which $abcd=0$.
Eqs. (\ref{param}) imply that the cases of $a=0$ or $d=0$ are now equivalent to
$b=(1/2)(N+\mathcal{S}_0)$ or $d=(1/2)(N-\mathcal{S}_0)$, respectively, in
(\ref{Poblet}). Whithin this framework, treating $\mathcal{S}_0$
as a parameter means that we intersect
the three-dimensional tetrahedron with hyperplanes of constant
$\mathcal{S}_0$ and look for the minimal value of $I_N^{e/o}$ for
points lying on the boundary of the resulting two-dimensional object.
Then, we choose the optimal solution among those parametrized by $\mathcal{S}_0$.

In what follows we will separate the proof into the
cases of even and odd $N$, and in each one we consider
all the facets of the tetrahedron separately.

\subsubsection{Even $N$}

When combining Eqs. (\ref{Penedes}) and (\ref{Besalu}) with Eq.
(\ref{BellEven}), one obtains a function in $b$ and $d$ parameterized by $\mathcal{S}_0$:
\begin{eqnarray}\label{Tarragona}
I_N^e(b,d;\mathcal{S}_0)&=&\frac{1}{2}\Big\{\frac{N(N-1)}{2}(\mathcal{S}
_0^2-N)-[N-2(b+d)]^2 \nonumber\\
&&+N\left [ \mathcal{S}_0(N-2(b+d)-1)+2(b-d) \right]+N\Big\}.\nonumber
\end{eqnarray}

\textit{Case a=0.} As commented before, this case is equivalent to
$b=(N+\mathcal{S}_0)/2$, which when applied to Eq. (\ref{Tarragona})
gives
\begin{eqnarray}\label{Badalona}
I_N^e(\tfrac{N+\mathcal{S}_0}{2},d;\mathcal{S}_0)&=&\frac{1}{4}
(N^2-3N-2)(\mathcal { S
} _0^2-N)\nonumber\\
&&-d[2d+\mathcal{S}_0(N+2)+N].
\end{eqnarray}
This is a quadratic function in $d$ which, since its second derivative with
respect to $d$ is negative, it has a local maximum.
Therefore, it attains its minimal value either at $d=0$ or
$d=(N-\mathcal{S}_0)/2$. In the first case,
$I_N^e(\tfrac{N+\mathcal{S}_0}{2},0;\mathcal{S}_0)=(1/4)(\mathcal{S}
_0^2-N)(N^2-3N-2 )$, which is clearly minimal for $\mathcal{S}_0=0$. Hence
$I_N^e(\tfrac{N+\mathcal{S}_0}{2},0;0)=-(1/4)N(N^2-3N-2 )$.

In the second case, i.e., $d=(N-\mathcal{S}_0)/2$,
Eq. (\ref{Badalona}) implies
\begin{equation}\label{ElPrat}
I_N^e(\tfrac{N+\mathcal{S}_0}{2},\tfrac{N-\mathcal{S}_0}{2};\mathcal{S}
_0)=-\frac { N
(N-1)}{4}[N+2-\mathcal{S}_0(\mathcal{S}_0-2)].
\end{equation}
It is easy to check that this expression attains its lowest value at either
$\mathcal{S}_0=0$ or $\mathcal{S}_0=2$, which corresponds to
\begin{equation}
I_N^e(\tfrac{N+\mathcal{S}_0}{2},\tfrac{N-\mathcal{S}_0}{2};\mathcal{S}
_0)=-\frac { 1 } { 4 } N(N-1)(N+2)\quad (\mathcal{S}_0=0,2),
\end{equation}
i.e., the value in Eq. (\ref{betaC}). Hence, for two different $4$-tuples
\begin{equation}\label{vec1}
\left(0,\frac{N}{2},0,\frac{N}{2}\right),\qquad
\left(0,\frac{N}{2}+1,0,\frac{N}{2}-1\right),
\end{equation}
we obtain the value of $I_N^e$ given in Eq. (\ref{betaC}).

\textit{Case b=0.} It follows from Eq. (\ref{Tarragona})
\begin{eqnarray}\label{PinkFloyd}
I_N^e(0,d;\mathcal{S}_0)&=&\frac{1}{2}\Big\{\frac{N(N-1)}{2}(\mathcal{S}
_0^2-N)-
(N-2d)^2+N\nonumber\\
&&\hspace{0.5cm}+N[\mathcal{S}_0(N-2d-1)-2d]\Big\}.
\end{eqnarray}
It is again not difficult to see that the second derivative of
$I_N^e(0,d;\mathcal{S}_0)$ with respect to $d$ is negative, and therefore we
look for its minimal value at the boundary of the range of $d$. For $d=0$
the above expression reduces to the right-hand side of
Eq. (\ref{ElPrat}), which, as previously mentioned, has minima for
$\mathcal{S}_0=0$ and $\mathcal{S}_0=2$. This results in two additional elements
of $\mathbbm{T}_4$ for which $I_N^e$ attains the value in Eq. (\ref{betaC}),
i.e.,
\begin{equation}\label{vec2}
\left(\frac{N}{2}-1,0,\frac{N}{2}+1,0\right),\qquad
\left(\frac{N}{2},0,\frac{N}{2},0\right)
\end{equation}
For $d=(N-\mathcal{S}_0)/2$, it follows from
Eq. (\ref{PinkFloyd}) that
\begin{equation}
I_N^e(0,\tfrac{N-\mathcal{S}_0}{2};\mathcal{S}_0)=\frac{1}{4}(\mathcal{S}
_0^2-N)(N^2+N-2),
\end{equation}
which has a minimum at $\mathcal{S}_0=0$ giving the fifth
point saturating our Bell inequality
\begin{equation}\label{vec3}
\left(\frac{N}{2},0,0,\frac{N}{2}\right).
\end{equation}

\textit{Cases c=0 or d=0.} First,
we notice that the case $c=0$ is equivalent to $d=(N-\mathcal{S}_0)/2$. Then,
following exactly the same reasoning as above, one straightforwardly finds that
assuming either $d=(N-\mathcal{S}_0)/2$ or $d=0$, the lowest value of $I_N^e$ is
$-(1/4)N(N-1)(N+2)$ and is obtained for the same five vectors as before, i.e.,
(\ref{vec1}), (\ref{vec2}), and (\ref{vec3}).

\subsubsection{Odd N}

By applying Eqs. (\ref{param}), (\ref{Penedes}),
and (\ref{Besalu}), we can turn $I_N$ into a two-variable function
%
\begin{eqnarray}\label{SantBoi}
I_N^o(b,d;\mathcal{S}_0)&=&\frac{1}{4}N(N-1)[\mathcal{S}_0(\mathcal{
S } _0+4)-N]-2(b^2+d^2)\nonumber\\
&&-b[4d-1+N(\mathcal{S}_0-2)]-d(N\mathcal{S}_0-1)\nonumber\\
\end{eqnarray}
parameterized by $\mathcal{S}_0$. Notice that since $N$ is odd,
$\mathcal{S}_0$ is also odd.

\textit{Case a=0.} This case is equivalent to $b=(N+\mathcal{S}_0)/2$,
and the latter when applied to Eq. (\ref{SantBoi}) gives
\begin{eqnarray}\label{Figueres}
I_N^o(\tfrac{N+\mathcal{S}_0}{2},d;\mathcal{S}_0)&=&-d[2d+\mathcal{S}
_0(N+2)+2N-1 ]\nonumber\\
&&+\frac{1}{4}(N^2-3N-2)(\mathcal{S}_0^2-N) \nonumber\\
&&+\frac { 1 } { 2 }
(N-1)^2\mathcal{S}_0\nonumber\\
&=&I_{N}^e(\tfrac{N+\mathcal{S}_0}{2},d;\mathcal{S}_0)+d+\frac{1}{2}
(N-1)^2\mathcal{S}_0.\nonumber\\
\end{eqnarray}
Since the second derivative of
$I_N^o(\tfrac{N+\mathcal{S}_0}{2},d;\mathcal{S}_0)$
with respect to $d$ is negative for any $\mathcal{S}_0$, it has its minimal
value either at $d=0$ or $d=(N-\mathcal{S}_0)/2$. For the first case, one
obtains
\begin{equation}
I^o_N(\tfrac{N+\mathcal{S}_0}{2},0;\mathcal{S}_0)=\frac{1}{4}
(\mathcal{S}_0^2-N)(N^2-3N-2)+\frac{1}{2}(N-1)^2\mathcal{S}
_0
\end{equation}
which is minimal at $\mathcal{S}_0=0$, and it reads
$I^o_N(\tfrac{N+\mathcal{S}_0}{2},0;0)=-(1/4)N(N^2-3N-2)$.

For $d=(N-\mathcal{S}_0)/2$, it follows from Eq. (\ref{Figueres}) that
\begin{equation}\label{lHospitalet}
I^o_N(\tfrac{N+\mathcal{S}_0}{2},\tfrac{N-\mathcal{S}_0}{2};\mathcal{S
} _0)=-\frac{1}{4}N(N-1)(N+4-\mathcal{S}_0^2).
\end{equation}
Clearly, the above expression is minimal for $\mathcal{S}_0=0$, but since
$\mathcal{S}_0$ must be odd, we obtain the lowest value for
$\mathcal{S}_0=\pm1$. As a result, the Bell expression
$I_N^o$ attains the minimum value (\ref{betaC}) at the following two
elements of $\mathbbm{T}_N$
\begin{equation}\label{ver1}
\left(0,\frac{N\pm1}{2},0,\frac{N\mp 1}{2}\right).
\end{equation}

\textit{Case b=0.} For $b=0$, Eq. (\ref{SantBoi}) simplifies to
\begin{eqnarray}\label{Castefa}
I^o_{N}(0,d;\mathcal{S}_0)&=&\frac{1}{4}N(N-1)[\mathcal{S}_0(\mathcal{
S } _0+4 )-N]-2d^2\nonumber\\
&&-d(N\mathcal{S}_0-1).
\end{eqnarray}
$I^o_{N}(0,d;\mathcal{S}_0)$ is thus a quadratic
function in $d$ that has a local maximum. Hence, it attains its minimal values
at the boundary of the range of $d$, i.e.,
either for $d=0$ or for $d=(N-\mathcal{S}_0)/2$. For $d=0$,
\begin{equation}
I^o_{N}(0,0;\mathcal{S}_0)=\frac{1}{4}N(N-1)[\mathcal{S}_0(\mathcal{
S } _0+4 )-N],
\end{equation}
which, since $\mathcal{S}_0$ must be odd, is minimal at
$\mathcal{S}_0=-3$ or $\mathcal{S}_0=-1$, and the corresponding value is
$-(1/4)N(N-1)(N+3)$. Consequently, at the following two vertices
\begin{equation}\label{ver2}
\left(\frac{N-1}{2},0,\frac{N+1}{2},0\right),\quad
\left(\frac{N-3}{2},0,\frac{N+3}{2},0\right)
\end{equation}
$I_N^o$ attains (\ref{betaC}).

For $d=(N-\mathcal{S}_0)/2$, Eq. (\ref{Castefa}) reads
\begin{equation}
I^o_{N}(0,\tfrac{N-\mathcal{S}_0}{2};\mathcal{S}_0)=-\frac{N-1}{4}
[ (N+2)(N-\mathcal{S}_0^2)-2(N+1)\mathcal{S}_0].
\end{equation}
Since $\mathcal{S}_0$ is odd, $I^o_{N}$ is minimal at
$\mathcal{S}_0=-1$. The corresponding minimum is
$-(1/4)N(N-1)(N+3)$, and it is attained at
\begin{equation}\label{ver3}
\left(\frac{N-1}{2},0,0,\frac{N+1}{2}\right).
\end{equation}

\textit{Cases c=0 or d=0.} Similarly, for either
$c=0$ or $d=0$, the minimal value of $I^o_N$ is
$-(1/4)N(N-1)(N+3)$, which is attained at the
five vertices (\ref{ver1}), (\ref{ver2}), and (\ref{ver3}).

In conclusion, the lowest value of $I_N^{e/o}$ for both even and odd $N$
is the one given in Eq. (\ref{betaC}) and is realized by five
elements of $\mathbbm{T}_N$: $(\ref{vec1})$, $(\ref{vec2})$,
and $(\ref{vec3})$ for even $N$, and (\ref{ver1}), (\ref{ver2}),
and (\ref{ver3}) for odd $N$. This also implies that the Bell inequality
$I_N^{e/o}$ is tangent to $\mathbbm{P}_{2,S}$ on five vertices. Since these are
linearly independent, the Bell inequality indeed represents a
facet of $\mathbbm{P}_2^S$.

\subsection{A bipartite reduction of the Dicke states}

The bipartite subsystem of the Dicke state $\ket{D_N^{k}}$ can be
determined analytically (see e.g. Ref. \cite{Molmer}).
For $k=\lceil N/2\rceil$, it has the following form
\begin{equation}
 \rho_{N}^{\lceil N/2\rceil}=\frac{1}{N(N-1)}
 \left(
 \begin{array}{cccc}
 p_N & 0 & 0 & 0 \\
 0 & q_N & q_N & 0 \\
 0 & q_N & q_N & 0 \\
 0 & 0 & 0 & r_N
\end{array}
\right),
 \end{equation}
where $p_N=(\lfloor N/2\rfloor -1)\lfloor N/2\rfloor $, $q_N=\lfloor N/2\rfloor \lceil N/2\rceil$, and $r_N=(\lceil N/2\rceil-1)\lceil N/2\rceil$.

\end{document}